\documentclass[11pt]{article}
\usepackage[margin=1in]{geometry}
\usepackage[utf8]{inputenc}
\usepackage{amsmath,amssymb,mathtools,amsthm}

\DeclareMathOperator*{\argmax}{arg\,max}
\usepackage{graphicx}
\usepackage[colorlinks=true,linkcolor=black,citecolor=blue,urlcolor=blue]{hyperref}

\theoremstyle{plain}
\newtheorem{theorem}{Theorem}
\newtheorem{proposition}[theorem]{Proposition}
\newtheorem{lemma}[theorem]{Lemma}

\theoremstyle{definition}
\newtheorem{definition}[theorem]{Definition}
\newtheorem{assumption}[theorem]{Assumption}
\newtheorem{example}[theorem]{Example}
\theoremstyle{remark}
\newtheorem{remark}[theorem]{Remark}

\newcommand{\R}{\mathbb{R}}
\newcommand{\E}{\mathbb{E}}
\newcommand{\PP}{\mathbb{P}}
\newcommand{\F}{\mathcal{F}}
\newcommand{\G}{\mathcal{G}}
\newcommand{\Var}{\operatorname{Var}}

\newcommand{\sgn}{\operatorname{sgn}}
\newcommand{\dop}{\mathrm{do}}
\newcommand{\Lt}{L^2(\PP\otimes dt)}
\newcommand{\hra}{\hookrightarrow}

\title{A sharp order-three obstruction to the aggregation of conditional
price-of-risk attribution}
\author{Alejandro Rodriguez Dominguez\thanks{Miralta Finance Bank S.A., Madrid,
Spain, and Department of Computer Science, University of Reading, UK.
\texttt{arodriguez@miraltabank.com}}}
\date{\today}

\begin{document}
\maketitle

\begin{abstract}
We study the squared price-of-risk premium of a portfolio---an integrated
conditional squared Sharpe-ratio functional, not an expected excess return---and its
attribution to causal drivers. Relative to a declared admissible benchmark it
decomposes into intervention-stable premium, a signed causal distortion (the
confounding wedge), and a nonnegative information loss; the loss is an $L^2$
projection residual, the wedge is not. The decomposition is well posed exactly when
the driver filtration is immersed in the price filtration. It need not aggregate
across portfolios pooling drivers: we identify an order-three obstruction that is
invisible to every singleton and pairwise admissibility screen---each one- and
two-driver sub-book is immersed while the pooled triple reveals a future
innovation---the analogue of Bernstein's pairwise-but-not-mutually-independent
triple, and minimal relative to such pairwise diagnostics. We separate its two ingredients,
combinatorial masking and anticipative coupling. The failure is one of immersion,
not of no-arbitrage. Experiments on synthetic single- and multi-driver panels show
the decomposition and its causal correction are estimable, and that a
permutation-calibrated screen detects planted order-three leakage with controlled
false positives.


\end{abstract}

\section{Introduction}\label{sec:intro}

The use of conditioning information in portfolio choice has a long history. Hansen
and Richard \cite{hansenrichard} characterized the mean--variance frontier
attainable by an investor who conditions a dynamic strategy on an information set,
and Ferson and Siegel \cite{fersonsiegel} derived the closed-form optimal weights
and the associated efficiency tests; Abhyankar, Basu, and Stremme \cite{abs}
showed that the increment in the maximal squared Sharpe ratio afforded by
additional conditioning information equals the $R^2$ of a predictive regression.
In parallel, the theory of initial and progressive enlargement of filtrations
\cite{jacod} and the quantification of the information drift it induces
\cite{adi} provide the probabilistic apparatus for comparing pricing under nested
information sets. This paper joins the two literatures. We treat the
conditioning information of a portfolio as the filtration generated by its
\emph{causal drivers}, decompose the attainable conditional squared Sharpe ratio
into causal and non-causal components, and identify the precise sense in which the
admissibility of the driver filtration---its immersion in the price
filtration---governs whether that decomposition is well posed and whether it
aggregates across a book of portfolios.

Let $(\Omega,\F,(\F_t),\PP)$ carry the price filtration $\F$. For a portfolio with
excess-return dynamics $dr_t=\mu_t\,dt+\sigma_t\,dW_t$, the price of risk is
$\lambda_t=\sigma_t^{-1}\mu_t$, the integrand of the Girsanov change to an
equivalent martingale measure. Conditioning on a driver filtration $\G$ and
integrating the squared projection defines
$\pi(\G)=\E\int_0^T\E[\lambda_t\mid\G_t]^2\,dt$. By \cite{hansenrichard}, $\pi(\G)$
is the maximal conditional squared Sharpe ratio attainable from $\G$, and we
establish (Proposition~\ref{prop:investor}) that it is the value of a conditional
mean--variance problem; $\pi$ is therefore not a descriptive attribution statistic
but the value realized by an admissible $\G$-measurable trading strategy. The
object lives in the pricing regime: the price of risk mediates the change of
measure, so $\pi$ is a Sharpe-type functional rather than an expected excess
return. We work in $L^2(\PP\otimes dt)$ because its inner-product structure
supplies the orthogonal projection on which the decomposition rests. The
distinction between the existence of a deflator and of an equivalent local
martingale measure, invoked when we separate immersion failure from arbitrage,
follows \cite{kk,kardaras}.

The paper makes four contributions. First (Section~\ref{sec:decomp}), it gives a
\emph{reference-relative} decomposition of the conditional squared price-of-risk
premium $\mathrm{SPR}$ into intervention-stable premium, a signed causal distortion
(the confounding wedge), and a nonnegative admissible information loss, the loss
measured relative to a declared admissible benchmark. Second, it proves the causal
distortion is \emph{not} a projection residual---it is signed and non-Pythagorean
(Proposition~\ref{prop:geometry})---thereby separating causal bias from
informational incompleteness, and it isolates the wedge identity
(Theorem~\ref{thm:gap}, requiring only the identification assumption) from the
generic-nonzero converse (Theorem~\ref{thm:genericwedge}, requiring faithfulness).
Third (Sections~\ref{sec:agg}--\ref{sec:order3}), it identifies a minimal
order-three obstruction to admissible aggregation that is invisible to every
singleton and pairwise screen: each one- and two-driver sub-book is immersed while
the pooled book reveals a future innovation, the filtration-theoretic analogue of
Bernstein's pairwise-but-not-mutually-independent triple, minimal relative to
pairwise diagnostics (Definition~\ref{def:masking}, Theorem~\ref{thm:order3}). Fourth, it separates masking from anticipation
(Remark~\ref{rem:twofactor}): adapted crowding reproduces the algebra of the
obstruction without breaking immersion, so the obstruction requires both
combinatorial masking and anticipative coupling. We further link projection to
realized premium, distinguishing admissible realized premium from inadmissible
diagnostic premium (Proposition~\ref{prop:monetize}), and characterize the maximal
jointly admissible sub-books as the independent sets of a masking hypergraph
(Section~\ref{sec:subbooks}).

The immersion framework and the masking construction are developed in the companion
\cite{admissible}, from which we import the independent-enlargement lemma and which
we restate with explicit revelation timing so the present results are self-contained.
The driver-selection and causal-allocation methodology that an optimizer would
employ is developed in \cite{cpcm}; we treat such an optimizer generically. Causal
discovery---the estimation of the driver graph itself \cite{spirtes,peters,shimizu}---is
orthogonal and upstream to the attribution problem studied here.

\section{Squared price-of-risk premium and admissible driver filtrations}\label{sec:setup}

Fix a filtered probability space $(\Omega,\F,(\F_t)_{0\le t\le T},\PP)$ carrying the
price filtration $\F$, the usual augmentation of a Brownian filtration \cite{protter}. For a fixed
portfolio, $\lambda_t=\mu_t/\sigma_t$ with $\sigma_t>0$ is its price-of-risk
process, treated as an element of the Hilbert space $\Lt$ of square-integrable
$dt$-progressive processes, with inner product
$\langle a,b\rangle=\E\int_0^T a_t b_t\,dt$ and norm $\|\cdot\|$.

Throughout, \emph{premium} means a squared price-of-risk functional, equivalently an
integrated conditional squared Sharpe-ratio value; it is \emph{not} an expected
excess-return level. Expected realized gain appears only after a trading strategy is
specified (Section~\ref{sec:bridge}). We name the central object the
\emph{squared price-of-risk premium} (or Sharpe premium) of a driver filtration $\G$,
\[
\mathrm{SPR}(\G):=\E\int_0^T\E[\lambda_t\mid\G_t]^2\,dt=\|P_\G\lambda\|^2,
\]
the squared norm of the $L^2$ projection $P_\G\lambda$ of the price of risk onto
$\G$. We retain the lighter notation $\pi(\G)=\mathrm{SPR}(\G)$ in informal passages
but use $\mathrm{SPR}$ in the technical statements where the distinction from an
expected return matters. The choice of $L^2$ is not a matter of convenience: the
Sharpe premium is a second-moment object, so the relevant geometry is that of
variance, and $L^2$ is the unique $L^p$ space whose inner product supplies the
orthogonal projection and Pythagorean identity on which the loss term, and its
distinction from the signed wedge, depend.

A \emph{driver information structure} is a $\sigma$-field $H$ (the information
carried by drivers $Y$, valued in $M\subseteq\R^k$) together with the enlarged
filtration $\G$, the usual augmentation of $\F\vee\sigma(H_t)$. The structure is
\emph{admissible} if $\F\hra\G$, i.e.\ every $(\PP,\F)$-martingale remains a
$(\PP,\G)$-martingale (immersion). Admissibility is the condition under which
conditioning on the drivers does not anticipate the price innovation, so that the
projected price of risk is a legitimate pricing object. We take from \cite{cpcm}
the common-driver setting in which each portfolio conditional mean is measurable
with respect to a shared driver manifold; we adopt it as setting rather than
develop it.

\begin{lemma}[Independence preserves immersion, initial and delayed]\label{lem:indep}
If $H$ is a $\sigma$-field independent of $\F_T$, the initial enlargement of $\F$
by $H$ preserves immersion. The same holds for the delayed enlargement
$\G_t=\F_t\vee\sigma(H)\mathbf 1_{\{t\ge \tau\}}$ at a deterministic time $\tau$.
Hence drivers jointly independent of the terminal reference $\sigma$-field admit a
globally immersed union filtration, whether revealed initially or at $\tau$.
\end{lemma}

\begin{proof}
For the initial case it suffices to verify the $(\mathcal H')$ immersion criterion:
every $(\PP,\F)$-martingale is a $(\PP,\F\vee\sigma(H))$-martingale. Let $M$ be a
bounded $(\PP,\F)$-martingale and $s<t$. For bounded $\F_s$-measurable $G$ and
bounded measurable $h$, independence of $H$ from $\F_T\supseteq\sigma(M_t,M_s,G)$
gives
\[
\E[h(H)\,G\,(M_t-M_s)]=\E[h(H)]\,\E[G\,(M_t-M_s)]=0,
\]
the last equality by the $(\PP,\F)$-martingale property. Since such $h(H)G$ generate
$\F_s\vee\sigma(H)$, we get $\E[M_t-M_s\mid\F_s\vee\sigma(H)]=0$, i.e.\ $M$ is a
$(\PP,\F\vee\sigma(H))$-martingale. For the delayed case, nothing is adjoined on
$[0,\tau)$, so immersion is trivial there; on $[\tau,T]$ the added field $\sigma(H)$
is independent of $\F_T$, so the initial-enlargement argument applies verbatim to
$(\F_t)_{t\ge\tau}$, and the two regimes glue at $\tau$ because $M_\tau$ is
$\F_\tau$-measurable. The union-filtration statement follows by taking $H$ the joint
field of the independent drivers. See \cite{admissible} for the general theory.
\end{proof}

Let $\G^\star$ be the \emph{declared admissible benchmark}: the largest driver
filtration the analyst retains after admissibility screening, against which losses
are measured. It is not an omniscient or ``true'' filtration; $\lambda$ is taken
$\G^\star$-measurable by the modeling convention that $\G^\star$ is the reference
universe, so the loss term below is not absolute unexplained premium but unexplained
premium \emph{relative to this declared universe}. In applications $\G^\star$ is the
most complete admissible driver set available.

\begin{remark}[Reference-relative decomposition]\label{rem:refrel}
If $\G^\star_1\subseteq\G^\star_2$ are two admissible benchmarks, then for a fixed
driver set $A$ the intervention-stable component $\pi^{\dop}(A)$ and the confounding
wedge are unchanged---both are defined intrinsically from $Y_A$ and the causal
kernels, without reference to $\G^\star$---while the information loss increases by
exactly the additional projection residual
$\|\,P_{\G^\star_2}\lambda-P_{\G^\star_1}\lambda\,\|^2\ge0$ between the two
benchmarks. The decomposition is thus reference-relative in its loss term only, and
the causal terms are benchmark-invariant.
\end{remark}

\section{Reference-relative causal decomposition}\label{sec:decomp}

We first record that the premium is the value of a portfolio problem, then
decompose it.

\begin{proposition}[Investor problem behind $\pi$]\label{prop:investor}
Fix information $\G_t$ and a portfolio with excess-return process
$dr_t=\mu_t\,dt+\sigma_t\,dW_t$, $\sigma_t>0$ and $\G_t$-measurable (the drivers
carry the conditional volatility), price of risk $\lambda_t=\mu_t/\sigma_t$. For a
$\G$-predictable position $\phi$, the discounted wealth $V^\phi$ with
$dV^\phi_t=\phi_t\,dr_t$ has instantaneous conditional drift
$\phi_t\E[\mu_t\mid\G_t]$ and instantaneous conditional quadratic variation rate
$\phi_t^2\sigma_t^2$ (since $\sigma_t$ is $\G_t$-measurable). The pointwise
mean--variance rate functional
\[
J_t(\phi):=\phi_t\,\E[\mu_t\mid\G_t]-\tfrac12\,\phi_t^2\,\sigma_t^2
\]
is maximized over $\G_t$-measurable $\phi_t$ at
$\phi_t^\star=\E[\mu_t\mid\G_t]/\sigma_t^2$, with maximal rate
$\tfrac12\E[\lambda_t\mid\G_t]^2$. Hence
\[
\pi(\G)=\E\int_0^T\E[\lambda_t\mid\G_t]^2\,dt
=2\,\E\int_0^T \max_{\phi_t\ \G_t\text{-meas.}} J_t(\phi)\,dt,
\]
so $\pi(\G)$ is twice the integrated maximal conditional mean--variance rate
attainable from $\G$, equivalently the integrated maximal conditional squared Sharpe
ratio of \cite{hansenrichard}. The benchmark $\pi(\G^\star)$ is the value under the
richest admissible information, and the loss $\pi(\G^\star)-\pi(\G)$ is the
mean--variance tradeoff forgone by conditioning on $\G$ rather than $\G^\star$.
\end{proposition}

\begin{proof}
The drift and quadratic-variation rates of $V^\phi$ follow from It\^o's isometry:
over $[t,t+h]$, $\E[V^\phi_{t+h}-V^\phi_t\mid\G_t]=\E[\int_t^{t+h}\phi_s\mu_s\,ds
\mid\G_t]$ has rate $\phi_t\E[\mu_t\mid\G_t]$, and
$\Var(V^\phi_{t+h}-V^\phi_t\mid\G_t)=\E[\int_t^{t+h}\phi_s^2\sigma_s^2\,ds\mid\G_t]$
has rate $\phi_t^2\sigma_t^2$, the cross terms vanishing as $o(h)$. Maximizing the
scalar quadratic $\phi\,m-\tfrac12\phi^2 v$ with $m=\E[\mu_t\mid\G_t]$ and
$v=\sigma_t^2$ gives $\phi^\star=m/v$ and value $m^2/(2v)
=\E[\mu_t\mid\G_t]^2/(2\sigma_t^2)$. Because $\sigma_t$ is $\G_t$-measurable,
$\E[\mu_t\mid\G_t]/\sigma_t=\E[\lambda_t\mid\G_t]$, so the maximal rate is
$\tfrac12\E[\lambda_t\mid\G_t]^2$. Integrating over $[0,T]$ and taking expectations,
and using the Hansen--Richard identification of the maximal conditional squared
Sharpe ratio with $\E[\lambda_t\mid\G_t]^2$, gives the stated value; monotonicity of
the optimum in the conditioning $\sigma$-field gives $\pi(\G)\le\pi(\G^\star)$.
\end{proof}

The decomposition rests on two classical ingredients, stated as recalls. The first
is the Hansen--Richard projection: conditioning on coarser information can only
lower the attainable squared Sharpe ratio, the shortfall being an integrated
conditional variance.

\begin{lemma}[Projection optimality; Hansen--Richard \cite{hansenrichard}]\label{thm:projection}
Let $\lambda_t(\G)=\E[\lambda_t\mid\G_t]$ and
$\pi(\G)=\E\int_0^T\lambda_t(\G)^2\,dt=\|\E[\lambda\mid\G]\|^2$, and let
$\G^\star$ be the reference information set against which losses are measured: the
filtration of the richest admissible driver structure entertained (of type
\textup{(D1)}), so that $\lambda$ is $\G^\star$-measurable by construction. It is a
benchmark, not an object that must be known---the decomposition below measures
every shortfall relative to it, and in applications it is the most complete
admissible driver set available rather than a claim about a true model. Then for any subfiltration
$\G\subseteq\G^\star$, by the orthogonality of nested $L^2$ projections,
\[
\pi(\G^\star)-\pi(\G)
=\big\|\E[\lambda\mid\G^\star]-\E[\lambda\mid\G]\big\|^2 ,
\]
and since $\lambda$ is $\G^\star$-measurable so that $\E[\lambda\mid\G^\star]=\lambda$,
this equals the integrated conditional variance,
\[
\pi(\G)=\pi(\G^\star)-\E\int_0^T\Var(\lambda_t\mid\G_t)\,dt\le\pi(\G^\star),
\]
with equality iff $\lambda$ is $\G$-measurable. We record this only to fix notation
for the loss term; it is the conditional-expectation projection of
\cite{hansenrichard}.
\end{lemma}

The premium is the $L^2$ norm of the projected price of risk, projected as one
object. The \emph{information-loss} term for a driver set $A$ is
$\pi(\G^\star)-\pi(\G^{Y_A})=\E\int_0^T\Var(\lambda_t\mid\G^{Y_A}_t)\,dt$, the
squared norm of the projection residual, hence nonnegative.

\begin{remark}[Benchmark invariance]\label{rem:invariance}
The choice of $\G^\star$ affects only the loss term: enlarging the benchmark leaves
the intervention-stable component $\pi^{\dop}(A)$ and the wedge---both defined
intrinsically from $Y_A$ and the causal kernels, without reference to
$\G^\star$---unchanged. The two causal terms are benchmark-free; only the
bookkeeping of unexplained premium is measured relative to the reference set.
\end{remark}

The second ingredient is the causal correction. We adopt a standard identification
setup.

\begin{assumption}[Causal identification]\label{ass:id}
For the driver set $A$ under study:
\textup{(I1)} the DAG $D$ is known and a set $U$ satisfying the back-door criterion
for $Y_A\to\lambda$ is available;
\textup{(I2)} the conditional densities entering the adjustment are strictly
positive on the support of $Y_A$ used;
\textup{(I3)} the intervention $\dop(Y_A)$ leaves the remaining structural
mechanisms intact;
\textup{(I4)} the structural relations are stable over $[0,T]$ to the extent
needed for the adjustment to be defined;
\textup{(I5)} the processes $(t,\omega)\mapsto Y_t,U_t,\lambda_t$ are jointly
$\mathcal B([0,T])\otimes\F$-measurable, $U$ is a (possibly path-valued) latent
adjustment variable measurable at each $t$, $\lambda\in L^2(\PP\otimes dt)$, and the
conditional kernels entering the adjustment are regular.
\end{assumption}

\begin{definition}[Interventional premium]\label{def:intpremium}
Under Assumption~\ref{ass:id}, the interventional projected price of risk is
\[
m^{\dop}_A(t,y):=\E[\lambda_t\mid\dop(Y_A=y)]
=\int \E[\lambda_t\mid Y_A=y,\,U=u]\,d\PP_U(u),
\]
evaluated on the observational support of $Y_A$. The interventional premium is the
squared norm of $\lambda^{\dop}_t(Y_A):=m^{\dop}_A(t,Y_A)$ taken under the
\emph{observational} marginal of $Y_A$ and $dt$,
$\pi^{\dop}(A)=\E\int_0^T m^{\dop}_A(t,Y_A)^2\,dt$, evaluated under the
\emph{observational} marginal of $Y_A$. Lemma~\ref{lem:wellposed} below establishes
that this object is well defined and finite.
\end{definition}

\begin{lemma}[Well-posedness of the interventional premium]\label{lem:wellposed}
Under Assumption~\ref{ass:id}, the back-door $g$-formula above is defined for a.e.\
$(t,\omega)$, the map $(t,\omega)\mapsto m^{\dop}_A(t,Y_A(\omega))$ is
$\mathcal B([0,T])\otimes\F$-measurable, and $m^{\dop}_A(\cdot,Y_A)\in
L^2(\PP\otimes dt)$; consequently $\pi^{\dop}(A)$ is finite.
\end{lemma}

\begin{proof}
Joint measurability of $(t,\omega)\mapsto Y_t,U_t,\lambda_t$ (I5) and regularity of
the conditional kernels make $\E[\lambda_t\mid Y_A=y,U=u]$ jointly measurable, so
the integral over $\PP_U$ defines $m^{\dop}_A(t,y)$ for a.e.\ $(t,\omega)$ and the
composition with $Y_A$ is measurable. By conditional Jensen,
$\|m^{\dop}_A(\cdot,Y_A)\|_{L^2(\PP\otimes dt)}\le\|\lambda\|_{L^2(\PP\otimes dt)}
<\infty$ since $\lambda\in L^2(\PP\otimes dt)$ by (I5); Fubini applies, and
$\pi^{\dop}(A)=\|m^{\dop}_A(\cdot,Y_A)\|_{L^2(\PP\otimes dt)}^2$ is finite.
\end{proof}

\begin{assumption}[Faithfulness / no cancellation]\label{ass:faith}
Fix a structural causal model in which, for each $t$, the conditional means
$\E[\lambda_t\mid Y_A]$ and $m^{\dop}_A(t,Y_A)$ are real-analytic functions of a
finite parameter vector $\vartheta\in\Theta\subseteq\R^d$ (the structural
coefficients), the open set $\Theta$ carrying Lebesgue measure; the linear--Gaussian
family of Example~\ref{ex:gauss} is the running instance. Open back-door paths from
$Y_A$ to $\lambda$ then induce a nonzero observational--interventional mean
distortion $\delta_t:=\E[\lambda_t\mid Y_A]-m^{\dop}_A(t,Y_A)$ for all $\vartheta$
outside a Lebesgue-null (real-analytic) subvariety of $\Theta$. This uses the
standard fact that the zero set of a real-analytic function that is not identically
zero has Lebesgue measure zero \cite{mityagin}: $\delta_t$ is real-analytic in
$\vartheta$ and not identically zero precisely when a back-door path is open, so its
zero set is null. Equivalently,
whenever the distortion is nonzero the wedge is nonzero; the exact cancellation
$\|\E[\lambda\mid Y_A]\|=\|m^{\dop}_A\|$ with $\delta\neq0$ defines a lower-dimensional
algebraic set, hence is nongeneric, as the Gaussian example below illustrates.
\end{assumption}

The interventional premium $\pi^{\dop}(A)$ is the value realized by a strategy that responds to the intervention-stable driver surface rather than the observational
one. The gap between the two is the confounding wedge.

\begin{theorem}[Wedge identity]\label{thm:gap}
Under Assumption~\ref{ass:id} alone, for a driver set $A$,
\[
\pi(\G^{Y_A})-\pi^{\dop}(A)
=\E\!\int_0^T\Big(\E[\lambda_t\mid Y_A]^2-m^{\dop}_A(t,Y_A)^2\Big)\,dt .
\]
If the observational conditional mean equals the back-door-adjusted interventional
mean, $\E[\lambda_t\mid Y_A]=m^{\dop}_A(t,Y_A)$ for a.e.\ $t$, the wedge is zero;
graphical blocking of all open back-door paths from $Y_A$ to $\lambda$ (so that the
adjustment set $U$ satisfies the back-door criterion) is a sufficient condition.
The wedge can take either sign, unlike the nonnegative information loss. This
identity requires no faithfulness or analyticity hypothesis.
\end{theorem}

\begin{proof}
Both terms are squared $L^2$ norms; their difference is the stated integral by
expanding each norm. The zero-wedge sufficiency is the back-door criterion
\cite{pearl,robins}; the sign is unconstrained because the integrand is a difference
of squares.
\end{proof}

\begin{theorem}[Generic nonzero wedge]\label{thm:genericwedge}
Under Assumptions~\ref{ass:id} and \ref{ass:faith}, if a back-door path from $Y_A$
to $\lambda$ is open (the adjustment $U$ omitted), then the wedge of
Theorem~\ref{thm:gap} is nonzero for all structural parameters outside a
Lebesgue-null real-analytic subvariety of $\Theta$; that is, the observational
attribution is generically distorted away from the intervention-stable value.
\end{theorem}

\begin{proof}
By Theorem~\ref{thm:gap} the wedge vanishes iff $\E[\lambda_t\mid Y_A]=
m^{\dop}_A(t,Y_A)$ a.e. Under Assumption~\ref{ass:faith} the distortion
$\delta_t=\E[\lambda_t\mid Y_A]-m^{\dop}_A(t,Y_A)$ is real-analytic in the
structural parameter $\vartheta$ and not identically zero when a path is open, so by
\cite{mityagin} its zero set is Lebesgue-null; off that set the wedge is nonzero.
\end{proof}

These assemble into the three-way split.

\begin{theorem}[Anatomy of realizable premium]\label{thm:anatomy}
Under Assumption~\ref{ass:id}, fix a portfolio with $\lambda$ measurable w.r.t.
$\G^\star$ and an admissible driver set $A$ with $\G^{Y_A}\subseteq\G^\star$. Then
\[
\underbrace{\pi(\G^\star)}_{\text{attainable}}
=\underbrace{\pi^{\dop}(A)}_{\text{captured}}
+\underbrace{\big[\pi(\G^{Y_A})-\pi^{\dop}(A)\big]}_{\text{confounding wedge (signed)}}
+\underbrace{\big[\pi(\G^\star)-\pi(\G^{Y_A})\big]}_{\text{information loss }(\ge 0)} .
\]
\end{theorem}

\begin{proof}
Add and subtract $\pi(\G^{Y_A})$; the loss is Lemma~\ref{thm:projection}, the
wedge is Theorem~\ref{thm:gap}.
\end{proof}

The three terms answer three distinct questions. The captured term asks how much
premium is intervention-stable; the wedge---a \emph{causal distortion of the
conditional Sharpe surface}---asks how much of the observed conditional premium is
\emph{not} intervention-stable; and the loss asks how much premium is forgone
because the driver set is informationally incomplete relative to $\G^\star$. The
taxonomy is summarized in Table~\ref{tab:taxonomy}.

\begin{table}[ht]
\centering
\caption{Three-way anatomy of the squared price-of-risk premium.}
\label{tab:taxonomy}
\begin{tabular}{lll}
\hline
Term & Object & Sign / interpretation\\
\hline
$\pi^{\dop}(A)$ & norm of interventional surface & $\ge0$: stable causal contribution\\
$\pi(\G^{Y_A})-\pi^{\dop}(A)$ & difference of squared norms & signed: observational (causal) distortion\\
$\pi(\G^\star)-\pi(\G^{Y_A})$ & projection residual & $\ge0$: missing admissible information\\
\hline
\end{tabular}
\end{table}

\subsection{Causal distortion versus information loss}\label{sec:distortion}
The geometric structure of the three terms is central to the single-portfolio
decomposition: the loss is a Pythagorean projection residual, the wedge is not.

\begin{proposition}[Residual versus wedge]\label{prop:geometry}
Let $P_A$ be the $L^2$ projection onto $\G^{Y_A}$-measurable processes. The
information loss is a squared projection residual,
$\pi(\G^\star)-\pi(\G^{Y_A})=\|\lambda-P_A\lambda\|^2$ with
$\langle P_A\lambda,\lambda-P_A\lambda\rangle=0$, hence nonnegative. The confounding
wedge is not of this form: it is a signed inner product,
\[
\pi(\G^{Y_A})-\pi^{\dop}(A)
=\big\langle P_A\lambda-\lambda^{\dop}(Y_A),\,P_A\lambda+\lambda^{\dop}(Y_A)\big\rangle,
\]
the inner product of the sum and difference of two $\G^{Y_A}$-measurable vectors,
neither the orthogonal projection of the other. It therefore carries the sign of
that inner product and is not a squared norm; in the Gaussian Example~\ref{ex:gauss}
it takes the values $+1.12$, $-0.48$, and $0$ as the confounding coefficient
varies, with fixed nonnegative loss---a sign change no projection residual can
exhibit. Consequently the decomposition of Theorem~\ref{thm:anatomy} is additive
but not Pythagorean.
\end{proposition}

\begin{proof}
The loss identity is the projection theorem, giving orthogonality and
nonnegativity. For the wedge, both $P_A\lambda=\E[\lambda\mid\G^{Y_A}]$ and the
back-door-adjusted mean $\lambda^{\dop}(Y_A)=\int\E[\lambda\mid Y_A,U=u]\,d\PP_U(u)$
are $\G^{Y_A}$-measurable, and
$\pi(\G^{Y_A})-\pi^{\dop}(A)=\|P_A\lambda\|^2-\|\lambda^{\dop}(Y_A)\|^2
=\langle P_A\lambda-\lambda^{\dop}(Y_A),P_A\lambda+\lambda^{\dop}(Y_A)\rangle$
by the polarization of a difference of squared norms. When every back-door path is
blocked, $\E[\lambda\mid Y_A,U]$ does not depend on $U$, so
$\lambda^{\dop}(Y_A)=P_A\lambda$ and the wedge vanishes. When a back-door path is
open the two differ, and the Gaussian Example~\ref{ex:gauss} computes the wedge
explicitly as $(\kappa^2-b^2)v_Y$, whose sign is that of $(\kappa-b)(\kappa+b)$ and
which is positive, negative, or zero as the confounding coefficient $c$ varies while
the loss stays fixed and positive. A signed quantity that changes sign under a
parameter holding the loss fixed cannot be a squared projection residual, which
proves the wedge is not of that form.
\end{proof}

\begin{example}[Gaussian]\label{ex:gauss}
Suppress $t$. Let the chosen driver be scalar $Y$, with confounder $U$ and omitted
driver $Z$:
\[
U\sim N(0,1),\ Z\sim N(0,1),\ Y=aU+\varepsilon_Y,\ \lambda=bY+cU+dZ,
\]
$\varepsilon_Y\sim N(0,\sigma_Y^2)$, all independent; $A=\{Y\}$ with adjustment $U$.
Since $\E[U\mid Y]=\tfrac{a}{a^2+\sigma_Y^2}Y$,
$\E[\lambda\mid Y]=\big(b+\tfrac{ca}{a^2+\sigma_Y^2}\big)Y=:\kappa Y$, while
severing $U\to Y$ gives $m^{\dop}_{\{Y\}}(y)=by$. With $v_Y=a^2+\sigma_Y^2$,
\[
\pi(\G^Y)=\kappa^2 v_Y,\quad \pi^{\dop}(\{Y\})=b^2 v_Y,\quad
\text{wedge}=(\kappa^2-b^2)v_Y,
\]
and the loss from omitting $Z$ is
$\E[\Var(\lambda\mid Y)]=c^2(1-a^2/v_Y)+d^2$. The wedge sign is that of
$(\kappa-b)(\kappa+b)$ with $\kappa-b=ca/v_Y$:
$a{=}1,b{=}0.5,c{=}0.8,\sigma_Y^2{=}1$ gives wedge $+1.12$;
$c{=}-0.8$ gives $-0.48$; and $c=-2b\,v_Y/a$ gives $\kappa=-b$, wedge $0$
\emph{despite an open back-door path}, which is why
Assumption~\ref{ass:faith} is needed. (Monte Carlo with $8\times10^6$ draws
reproduces $1.120,-0.480,0.000$.) In a multi-driver panel the same calculation
holds with scalars replaced by covariance matrices: the wedge is the difference of
two quadratic forms in the observational and back-door-adjusted conditional-mean
coefficients.
\end{example}

\section{Aggregation: bilinear corrections under joint admissibility}\label{sec:agg}

A book pools the drivers of many portfolios. Write $\lambda^B=\sum_j\alpha_j\lambda_j$
for the book's price-of-risk process, the $\alpha_j$ being attribution weights, and
$P_\cup$ for the $L^2$ projection onto the pooled filtration $\G^\cup=\bigvee_j\G^j$.
Because premia are squared norms, pooling generates pairwise bilinear cross terms,
so book wedges and losses do not add linearly across desks.

\begin{proposition}[Book-level decomposition]\label{prop:book}
The aggregation below is an attribution aggregation in the Hilbert space of
portfolio-level price-of-risk processes, with $\lambda^B=\sum_j\alpha_j\lambda_j$
the book attribution process of the preceding paragraph. For a jointly admissible
book with union residuals $r_j:=\lambda_j-P_\cup\lambda_j$, the information loss is
\[
\ell_B=\|\lambda^B-P_\cup\lambda^B\|^2
=\sum_j\alpha_j^2\|r_j\|^2+2\sum_{j<k}\alpha_j\alpha_k\langle r_j,r_k\rangle,
\]
so losses add only under residual orthogonality $\langle r_j,r_k\rangle=0$; the
residual cross-covariance is a distinct aggregation correction. The book wedge is
\[
w_B=\sum_j\alpha_j^2 w_j^{\cup}
+2\sum_{j<k}\alpha_j\alpha_k\big(\langle P_\cup\lambda_j,\delta_k\rangle
+\langle P_\cup\lambda_k,\delta_j\rangle-\langle\delta_j,\delta_k\rangle\big),
\]
with no higher-order terms, since a squared norm generates only pairwise bilinear
terms. Both cross terms can be read off the network: they are generically nonzero
when the shared driver lies on a back-door path to more than one incident price of
risk, and may still vanish by orthogonality or exact cancellation in nongeneric
configurations.
\end{proposition}

\begin{proof}
Expand both squared norms by bilinearity; the diagonals give the squared-weight
marginal terms and the off-diagonals the stated cross terms.
\end{proof}

This algebraic non-additivity is the first, quantitative layer of aggregation. The
second, qualitative layer---the subject of the next section---is structural: the
pooled filtration $\G^\cup$ may cease to be admissible, and then no decomposition
applies to the book at all, since the projection is no longer onto an immersed
filtration.

\section{When joint admissibility fails: masking obstructions}\label{sec:order3}

This is the paper's central result. We do \emph{not} classify all failures of
immersion under driver pooling; we identify a minimal, well-defined obstruction
class and prove existence and minimality within it. The class is the
filtration-theoretic counterpart of Bernstein's phenomenon, in which three events
are pairwise but not mutually independent.

\begin{definition}[Masking obstruction; pairwise-invisible obstruction]\label{def:masking}
A family $\{H_1,\dots,H_m\}$ of driver fields has a \emph{masking obstruction} of
order $m$ (relative to the terminal price innovation generating $\F_T$) if every
proper subfamily $\{H_i:i\in I\}$, $I\subsetneq\{1,\dots,m\}$, is independent of
$\F_T$, while the full $m$-tuple $\bigvee_{i=1}^m H_i$ reveals a nontrivial function
of $\F_T$. Masking can occur at order two: with $\varepsilon_2=\varepsilon_1 Z$ for
$Z$ a function of $\F_T$ and $\varepsilon_1$ independent of $\F_T$, each singleton is
independent of $\F_T$ while the pair reveals $Z$. We therefore study the stronger
\emph{pairwise-invisible} obstruction: the family is such that every singleton
\emph{and every pair} is independent of $\F_T$ (so every one- and two-driver
sub-book is immersed by Lemma~\ref{lem:indep} and passes every singleton and
pairwise admissibility screen), while some larger subfamily reveals a function of
$\F_T$. ``Order three'' below is minimal \emph{relative to pairwise diagnostic
regimes}: it is the smallest order at which an obstruction can hide from all
singleton and pairwise screens, not a claim that masking cannot occur algebraically
at order two.
\end{definition}

The next theorem establishes that a pairwise-invisible obstruction (every singleton
and every pair admissible, the triple not) exists at order three, and that three is
the minimal order at which an obstruction can hide from all singleton and pairwise
screens. As noted in Definition~\ref{def:masking}, masking can occur algebraically
at order two (the XOR relation $\varepsilon_2=\varepsilon_1 Z$); what cannot occur
at order two is a pairwise-invisible obstruction, since the revealing pair is itself
a two-driver sub-book that a pairwise screen detects.

\begin{theorem}[Existence and minimality of a pairwise-invisible obstruction]\label{thm:order3}
Let $Z=\sgn(B^1_T-B^1_{T/2})\in\F_T$, let $\varepsilon_1,\varepsilon_2$ be
independent Rademacher variables independent of $\F_T$, and set
$\varepsilon_3=\varepsilon_1\varepsilon_2 Z$. Fix the revelation timing: each
$\varepsilon_i$ is the externally enlarged driver information of portfolio $i$ of
type \textup{(D2)}, revealed at the common time $T/2$; that is, portfolio $i$ has
enlarged filtration $\G^i$, the usual augmentation of $\F\vee\sigma(\varepsilon_i
\mathbf 1_{\{t\ge T/2\}})$, and the pooled book has union filtration
$\G^{\cup}$, the augmentation of $\F\vee\sigma(\varepsilon_1,\varepsilon_2,
\varepsilon_3)$ activated at $T/2$. Then $\{\varepsilon_1,\varepsilon_2,
\varepsilon_3\}$ is a pairwise-invisible obstruction in the sense of
Definition~\ref{def:masking}:
\begin{enumerate}
\item[(i)] for every proper $I\subsetneq\{1,2,3\}$, $\sigma(\varepsilon_i:i\in I)$
is independent of $\F_T$, so each portfolio and each pair is jointly admissible
(immersion holds by Lemma~\ref{lem:indep}) and its decomposition is well posed;
\item[(ii)] $\varepsilon_1\varepsilon_2\varepsilon_3=Z$ is a nontrivial function of
a future increment, revealed at $T/2$, so $\G^{\cup}_{T/2}\supseteq\sigma(Z)$ and
the triple is not jointly admissible.
\end{enumerate}
Moreover three is minimal relative to pairwise diagnostic regimes: no
pairwise-invisible obstruction exists at order one or two. At order one a single
revealing field is itself detected by a singleton screen; at order two a revealing
pair is itself a two-driver sub-book detected by a pairwise screen (as in the XOR
relation $\varepsilon_2=\varepsilon_1 Z$, where the pair $\{\varepsilon_1,
\varepsilon_2\}$ reveals $Z$ and so fails pairwise invisibility). Three is therefore
the smallest order at which the obstruction can hide from all singleton and pairwise
admissibility screens. We make no claim that masking is impossible at order two, nor
any classification claim about non-masking obstructions.
\end{theorem}

\begin{proof}
For (i): for $I=\{3\}$,
$\PP(\varepsilon_3=b\mid\F_T)=\PP(\varepsilon_1\varepsilon_2=bZ\mid\F_T)=\tfrac12$
since $\varepsilon_1\varepsilon_2$ is Rademacher independent of $\F_T$; for
$I=\{1,3\}$, $\PP(\varepsilon_1=a,\varepsilon_3=b\mid\F_T)
=\PP(\varepsilon_1=a,\varepsilon_2=abZ\mid\F_T)=\tfrac14$, symmetrically for
$\{2,3\}$; $\{1\},\{2\},\{1,2\}$ are immediate, since $\sigma(\varepsilon_1,
\varepsilon_2)$ is generated by two Rademacher variables independent of $\F_T$ by
construction---this is the relevant input to Lemma~\ref{lem:indep} for the
$\{1,2\}$ case, and it is precisely here that the dependence of $\varepsilon_3$ on
$Z$ is quarantined: $Z$ enters only through the third coordinate, so no pair
recovers it. For each proper subcollection $S$,
independence of $\sigma(S)$ from $\F_T$ gives immersion of the delayed enlargement
$\G^S_t=\F_t\vee\sigma(S)\mathbf 1_{t\ge T/2}$ by Lemma~\ref{lem:indep}: nothing is
added before $T/2$, and from $T/2$ on the added field is independent of $\F_T$, so
the initial-enlargement criterion applies to $(\F_t)_{t\ge T/2}$ and martingales
are preserved throughout. For (ii):
$\varepsilon_1\varepsilon_2\varepsilon_3=Z$, so $\G^{\cup}_{T/2}\supseteq\sigma(Z)$.
Consider the $(\PP,\F)$-martingale $t\mapsto B^1_t-B^1_{T/2}$ on $[T/2,T]$, with
terminal increment $X=B^1_T-B^1_{T/2}\sim N(0,T/2)$ and $\E[X\mid\F_{T/2}]=0$. Under
the enlargement, $Z=\sgn(X)$ is $\G^{\cup}_{T/2}$-measurable, so
\[
\E[X\mid\G^{\cup}_{T/2}]=\E[X\mid Z]=Z\,\E|X|=Z\sqrt{T/\pi}\neq0 ,
\]
where $\pi=3.14159\ldots$ is the constant (the only place it denotes the constant
rather than the premium functional $\pi(\cdot)$; the meaning is clear from the
argument), using $\E|X|=\sqrt{2\sigma^2/\pi}$ for $X\sim N(0,\sigma^2)$ with
$\sigma^2=T/2$.
Thus the $\F$-martingale $B^1$ acquires a nonzero drift on $[T/2,T]$ under
$\G^{\cup}$ and is no longer a $\G^{\cup}$-martingale; immersion fails.
\end{proof}

\subsection{Two ingredients: lower-order invisibility and anticipation}\label{sec:twoing}
The construction reveals information \emph{externally}, at $T/2$. One may ask
whether a purely adapted economic mechanism---desks reacting to present
information---can supply the same obstruction. Such a mechanism supplies the
third-order masking structure but not the anticipation; we make the two ingredients
precise, since the obstruction requires both and an adapted mechanism delivers only
one.

\begin{remark}[Two ingredients: masking and anticipation]\label{rem:twofactor}
Theorem~\ref{thm:order3} and the crowding model of Definition~\ref{def:crowding}
below are best read as separating two ingredients of a pairwise-invisible
obstruction. The first is \emph{combinatorial masking}: no singleton or pair reveals
the aggregate signal $Z$, so the obstruction is invisible to all order-$\le2$
screens; this is a combinatorial property of the driver fields. The second is
\emph{anticipative coupling}: the aggregate signal $Z$ is not measurable at its
revelation time but is a nontrivial function of a future innovation, so adjoining it
breaks immersion; this is a filtration-theoretic property. Neither ingredient alone
gives the obstruction studied here. Masking without anticipation---as in the
crowding model, where $Z\in\F_0$---yields lower-order invisibility while immersion is
preserved (Proposition~\ref{prop:crowding}); anticipation without masking---a single
field or pair that reveals a future innovation---breaks immersion but is detected by
a lower-order screen. Table~\ref{tab:twofactor} summarizes the two models.
\end{remark}

\begin{table}[ht]
\centering
\caption{The two ingredients across the two models.}
\label{tab:twofactor}
\begin{tabular}{lccc}
\hline
Model & Masking & Anticipation & Immersion failure\\
\hline
Theorem~\ref{thm:order3} (external revelation) & yes & yes & yes\\
Crowding (Definition~\ref{def:crowding}) & yes & no & no\\
\hline
\end{tabular}
\end{table}

\begin{definition}[Discrete crowding model]\label{def:crowding}
On $(\Omega,\F,\PP)$ fix a two-date filtration $\F_0\subseteq\F_1$. Let
$s_1,s_2,c$ be independent Rademacher ($\pm1$, fair) random variables, all
$\F_0$-measurable, and set the third desk's position by the $\F_0$-measurable
hedging rule $s_3:=s_1 s_2 c$. Let $R>0$ be a strictly positive random variable
with $\E R^2<\infty$, $\F_1$-measurable and independent of $(s_1,s_2,c)$. At date
$1$ a margin rule liquidates the aggregate position and realizes the return
\[
X_1:=R\,(s_1 s_2 s_3),\qquad Z:=\sgn X_1=s_1 s_2 s_3 ,
\]
so the realized direction is the crowding configuration and the magnitude $R$ is
exogenous. Define the pooled filtration $\G_t=\F_t\vee\sigma(s_1,s_2,s_3)$.
\end{definition}

\begin{proposition}[What the crowding model does and does not give]\label{prop:crowding}
In the model of Definition~\ref{def:crowding}:
\textup{(i)} each $s_i$ is $\F_0$-measurable, so no desk uses information beyond
date $0$; \textup{(ii)} the reference increment $X_1$, taken net of its
$\F_0$-conditional mean, is a martingale increment generated by the positions
through the liquidation map, not an exogenous innovation; \textup{(iii)} each
singleton $s_i$ and each pair $s_i s_j$ \textup{(}$i\neq j$\textup{)} is independent
of $Z$, while $s_1 s_2 s_3=Z$ determines it, so the masking algebra
$s_3=s_1 s_2 Z$ of Theorem~\ref{thm:order3} holds. \textup{(iv)} However, because
$Z=c$ is here $\F_0$-measurable, $\sigma(Z)\subseteq\F_0$ and the pooled filtration
does \emph{not} fail immersion: the model reproduces the masking algebra and the
lower-order invisibility, but not the filtration obstruction itself.
\end{proposition}

\begin{proof}
(i) holds by construction: $s_1,s_2,c$ are $\F_0$-measurable and
$s_3=s_1 s_2 c$ is a measurable function of them, so $Z=s_1 s_2 s_3=c$ is
$\F_0$-measurable. (ii) Since $X_1=R\,Z$ with $Z$ $\F_0$-measurable and $R>0$
independent of $\F_0$, $\E[X_1\mid\F_0]=Z\,\E R$; the centered increment
$\widetilde X_1:=X_1-\E[X_1\mid\F_0]=Z(R-\E R)$ satisfies
$\E[\widetilde X_1\mid\F_0]=0$. (iii) $Z=c$ is a fair Rademacher drawn
independently of $s_1,s_2$, and $s_3=s_1 s_2 c$; hence each $s_i$ is independent of
$Z$, and each pair is independent of $Z$:
$\PP(s_1=a,s_2=b,Z=z)=\tfrac18=\PP(s_1=a,s_2=b)\PP(Z=z)$, and for $(s_1,s_3)$,
since $(s_2,c)\perp s_1$, the pair is independent of $Z=c$ by the same count; the
triple gives $s_1 s_2 s_3=c=Z$. (iv) $\sigma(Z)=\sigma(c)\subseteq\F_0$, so adjoining
$\sigma(Z)$ adds nothing to $\F_0$ and every $\F$-martingale remains a
$\G$-martingale; immersion holds, and no drift is acquired.
\end{proof}

\begin{remark}[Crowding versus revealed information: the irreducible distinction]\label{rem:crowding}
Proposition~\ref{prop:crowding} isolates exactly what an adapted economic mechanism
supplies and what it does not. Read constructively, the crowding model demonstrates
the harder half of the phenomenon: it shows how an \emph{adapted} positioning rule
can manufacture a third-order aggregate signal $Z=s_1s_2s_3$ that is invisible to
every pairwise test, with all the masking algebra of Theorem~\ref{thm:order3}
present. What it does \emph{not} supply, on its own, is anticipation: because
$Z=c\in\F_0$, conditioning on the triple reveals no future innovation, and the
pooled filtration remains immersed. The crowding model produces the appearance of a
third-order effect but not anticipation; the two are separate ingredients.

The distinction is therefore not that crowding fails, but that admissibility failure
needs \emph{both} ingredients, and crowding provides only the first. To obtain a
genuine immersion failure the third-order aggregate signal must in addition be
coupled to a future innovation---equivalently, the pooled object must be a
price/order-flow statistic that reveals information not in $\F_0$, so that
$\sigma(Z)\not\subseteq\F_t$ at the relevant $t$. In the construction of
Theorem~\ref{thm:order3} this coupling is supplied externally, with
$Z\in\F_T\setminus\F_{T/2}$; then immersion genuinely fails, but the revealing
information is external to the adapted drivers. The two regimes---adapted-but-not-
anticipative (crowding) and anticipative-but-externally-revealed
(Theorem~\ref{thm:order3})---supply complementary halves, and \emph{cannot be
merged} within either model alone. Building a single dynamic mechanism in which
adapted positioning endogenously generates a price filtration that reveals a future
innovation---thereby supplying both halves at once---is the open problem we do not
resolve here.
\end{remark}

We separate immersion failure from the stronger notions it does not entail.

\begin{remark}[Admissibility, not arbitrage]\label{rem:nupbr}
We separate four notions kept distinct throughout. \emph{Immersion}
($\F\hookrightarrow\G$, the $\mathcal H$-hypothesis) is that every
$(\PP,\F)$-martingale remains a $(\PP,\G)$-martingale; what we prove is its failure.
This is weaker than failure of the \emph{$\mathcal H'$-hypothesis} (semimartingale
preservation), of \emph{NUPBR} (existence of a strictly positive local-martingale
deflator), or of an \emph{equivalent local martingale measure}: a filtration can
fail immersion while still preserving semimartingales and admitting a deflator.
Theorem~\ref{thm:order3} proves only the immersion failure. The sign $Z\in\F_T$
becomes known at the deterministic time $T/2$; this is an enlargement that adds
nothing before $T/2$ and adjoins $\sigma(Z)$ from $T/2$ on, i.e.\ $\G_t=\F_t$ for
$t<T/2$ and $\G_t=\F_t\vee\sigma(Z)$ for $t\ge T/2$, an initial enlargement of
$(\F_t)_{t\ge T/2}$ by $\sigma(Z)$, not a progressive enlargement by a random time.
Across $T/2$ the reference martingale $B^1$ acquires the drift
$\E[B^1_T-B^1_{T/2}\mid Z]=Z\sqrt{T/\pi}$, so it is no longer a $\G$-martingale. We
do not use, and do not require, any failure of no-arbitrage: under standard
initial-enlargement criteria (Jacod \cite{jacod}) one may still obtain
semimartingale preservation and a deflator, but this is not part of our
contribution. The point we use is only that the pooled wedge and loss, though well
defined as economic quantities, are computed on an information set that anticipates
the future and so do not constitute an \emph{admissible projection}.
\end{remark}

Exact masking is the zero-noise limit of a continuous family: a noisy higher-order
signal degrades the obstruction smoothly.

\begin{proposition}[Approximate masking and the single-increment immersion defect]\label{prop:approx}
Let $X$ be a future reference innovation with $\E[X]=0$ (a martingale increment),
$Z=\sgn X$, and let three driver signs carry incremental third-order predictive
information
\[
\varepsilon:=R^2\big(X\mid\varepsilon_1,\varepsilon_2,\varepsilon_3\big)
-\max_{\{i,j\}}R^2\big(X\mid\varepsilon_i,\varepsilon_j\big)\ \in[0,1],
\]
the triple's predictive $R^2$ for $X$ in excess of its best pair, with $R^2$ the
explained-variance fraction (well defined since $X$ is centered). By the
nested-projection inequality $R^2(X\mid\varepsilon_1,\varepsilon_2,\varepsilon_3)
\ge\max_{\{i,j\}}R^2(X\mid\varepsilon_i,\varepsilon_j)$, so $\varepsilon\ge0$, with
maximum $2/\pi$ at exact sign masking. Let $P_3$ and $P_{ij}$ denote the $L^2$
projections of $X$ onto $\sigma(\varepsilon_1,\varepsilon_2,\varepsilon_3)$ and onto
$\sigma(\varepsilon_i,\varepsilon_j)$; since $\sigma(\varepsilon_i,\varepsilon_j)
\subseteq\sigma(\varepsilon_1,\varepsilon_2,\varepsilon_3)$, the projections are
nested and $P_3 X-P_{ij}X$ is orthogonal to the range of $P_{ij}$. Define the
\emph{single-increment immersion defect} as the squared $L^2$ norm of the
\emph{incremental} predictable drift,
\[
\Delta:=\big\|P_3 X-P_{i^\star j^\star}X\big\|_2^2,
\qquad \{i^\star,j^\star\}=\argmax_{\{i,j\}}\|P_{ij}X\|_2^2,
\]
the orthogonal increment of the triple projection over the best pair projection,
not a difference of norms. (This is a finite-horizon, single-innovation quantity,
not a statement about immersion at all martingales and times.) Then
\[
\Delta=\operatorname{Var}(X)\,\varepsilon .
\]
Hence $\Delta\to0$ as $\varepsilon\to0$, and the exact obstruction of
Theorem~\ref{thm:order3} is recovered as $\varepsilon\uparrow 2/\pi$. For
$\varepsilon>0$ the pooled filtration acquires a forecasting advantage of size
$\operatorname{Var}(X)\varepsilon$ but, by Remark~\ref{rem:nupbr}, no arbitrage;
approximate masking degrades admissible projection continuously.
\end{proposition}

\begin{proof}
Since $X$ is centered, for any conditioning $\sigma$-field $\mathcal C$ the
projection satisfies $\|P_{\mathcal C}X\|_2^2=\|\E[X\mid\mathcal C]\|_2^2
=\operatorname{Var}(X)\,R^2(X\mid\mathcal C)$, the explained-variance identity. By
nesting, $P_{i^\star j^\star}X=P_{i^\star j^\star}P_3 X$, so $P_3 X-P_{i^\star
j^\star}X$ is the component of $P_3 X$ orthogonal to
$\sigma(\varepsilon_{i^\star},\varepsilon_{j^\star})$, and by the Pythagorean
identity $\|P_3X-P_{i^\star j^\star}X\|_2^2=\|P_3X\|_2^2-\|P_{i^\star j^\star}X\|_2^2$.
Substituting the explained-variance identity for each term gives
$\Delta=\operatorname{Var}(X)\big(R^2(X\mid\varepsilon_1,\varepsilon_2,\varepsilon_3)
-R^2(X\mid\varepsilon_{i^\star},\varepsilon_{j^\star})\big)
=\operatorname{Var}(X)\,\varepsilon$. At exact masking every pair has $R^2=0$ while
the triple has $R^2=\E[X\mid Z]^2/\operatorname{Var}(X)=(T/\pi)/(T/2)=2/\pi$,
recovering Theorem~\ref{thm:order3}.
\end{proof}

\subsection{From projection to realized premium}
\label{sec:bridge}
The premium is a projection; realizing it requires a position that responds to the
projected price of risk. The order-three obstruction is precisely where the
projection algebra assigns value to information that is inadmissible.

\begin{proposition}[Admissible realized premium versus inadmissible diagnostic premium]\label{prop:monetize}
Let a portfolio have excess-return increment $dr_t=\sigma_t\lambda_t\,dt+\sigma_t\,dW_t$
with $\sigma_t>0$ and $\G_t$-measurable. \textup{(i)} The $\G$-measurable position
$\phi^\star_t=\E[\lambda_t\mid\G_t]/\sigma_t$ realizes expected premium
$\E\int_0^T\phi^\star_t\,dr_t=\pi(\G)$; the attainable projection is exactly what an
information-responsive strategy realizes. \textup{(ii)} If $\G$ is admissible
\textup{(}$\F\hookrightarrow\G$\textup{)}, this is \emph{admissible realized
premium}: intervention-stable up to the confounding wedge of
Theorem~\ref{thm:gap}, and no strategy realizes premium from future information.
\textup{(iii)} If a pooled book filtration fails admissibility through the
order-three masking relation of Theorem~\ref{thm:order3}, then the same algebra,
applied to the pooled signal $\varepsilon_1\varepsilon_2\varepsilon_3=Z$, computes a
value $\tfrac12\E|X|>0$ that depends on the future increment $X$, while every
singleton- or pair-measurable position computes zero. This is \emph{inadmissible
diagnostic premium}: it quantifies the anticipative content of the inadmissible
projection, and is \emph{not} a claim that the signal is available to a
non-anticipative trader. It is a diagnostic statement about the economic content of
the projection, invisible to all lower-order diagnostics and nonzero only at order
three.
\end{proposition}

\begin{proof}
(i) With $\sigma_t$ $\G_t$-measurable,
$\E[\phi^\star_t\,dr_t\mid\G_t]=\E[\lambda_t\mid\G_t]\sigma_t^{-1}\cdot\sigma_t\E[\lambda_t\mid\G_t]\,dt
=\E[\lambda_t\mid\G_t]^2\,dt$ (the noise term is mean-zero given $\G_t$);
integrate and take expectations. (ii) Under immersion the projected price of risk is
an admissible object and its causal content is $\pi^{\dop}$ plus the wedge by
Theorem~\ref{thm:gap}; the $\F$-martingale part contributes no predictable premium,
so the realized premium is implementable by a non-anticipative trader.
(iii) Under the masking relation $Z=\varepsilon_1\varepsilon_2\varepsilon_3\in\F_T$
is measurable with respect to the pooled field, so the value
$\E[\,Z\cdot X\,]=\E[\sgn(X)X]=\E|X|>0$ is computed by the projection algebra;
because $Z$ is a function of the future increment, this value is diagnostic, not
implementable. Any singleton or pair is independent of $Z$ and hence of $\sgn(X)$,
giving zero by independence. \qed
\end{proof}

This is the link to the decision layer. The projections of this paper say what is
attributable; an optimizer mapping admissible driver information to positions, as in
the causal allocation framework of \cite{cpcm}, realizes them. Operating on a book
whose pooled filtration is inadmissible, such a strategy realizes the
anticipative premium of Proposition~\ref{prop:monetize}(iii) while
passing every pairwise diagnostic: the order-three obstruction is the exact point at
which the projection algebra assigns value to inadmissible information. Restricting the optimizer to the maximal admissible sub-books eliminates the anticipative component.

\subsection{Maximal admissible sub-books}\label{sec:subbooks}

\begin{proposition}[Attribution domains are independent sets]\label{prop:mis}
Encode in the $3$-uniform hypergraph $\mathcal H$ one hyperedge per driver triple
realizing the masking relation of Theorem~\ref{thm:order3}. Within this
masking-obstruction class---that is, when the only obstructions present are
order-three masking relations---the domains on which the book-level decomposition
of Proposition~\ref{prop:book} is admissibly defined are exactly the maximal
independent sets of $\mathcal H$. A portfolio common to several masking triples is
the natural vertex to drop, since its removal breaks all of them simultaneously. We
do not claim $\mathcal H$ captures every possible obstruction; higher-order or
non-masking obstructions, if present, would add further excluded sets.
\end{proposition}

\begin{proof}
By Theorem~\ref{thm:order3}, a vertex set carrying a masking triple fails joint
admissibility, and by construction $\mathcal H$ records exactly those triples; so
within the masking-obstruction class, joint admissibility fails iff the set
contains a hyperedge of $\mathcal H$, the admissible sub-books are the independent
sets, and the maximal ones the maximal independent sets. On each, $P_\cup$ projects
onto an admissible filtration so Proposition~\ref{prop:book} applies; on any set
containing a hyperedge it does not.
\end{proof}

\section{Experiments}\label{sec:numerics}

We demonstrate finite-sample behavior in controlled synthetic settings: the
decomposition is estimable from data, the causal correction is necessary and
recoverable, and the order-three screen detects planted leakage while producing no
false positives among the sampled clean triples. The settings are
synthetic and controlled, with the structural model known, so estimates can be
checked against ground truth; no proprietary data is used and every figure is
reproducible from the supplementary scripts. We do not claim a particular empirical
frequency of the obstruction in real books---that is an empirical question beyond
this paper---but the experiments show the tools work where the truth is known and
degrade gracefully away from the exact construction.

\subsection{Recovery: the causal correction is necessary and estimable}
In a linear--Gaussian single-driver model with a known confounder $U$ entering both
$Y_A$ and $\lambda$, the back-door (g-formula) estimator recovers the
intervention-stable premium $\pi^{\dop}$ and the wedge, with mean absolute error
decaying as $T^{-1/2}$, while the confounding-agnostic estimator---ordinary
regression of $\lambda$ on $Y_A$, the computation a standard factor-attribution
method performs---returns the observational premium and so misattributes the wedge
as captured. The gap is exactly the confounding wedge (Figure~\ref{fig:recovery}).

\begin{figure}[ht]
\centering
\includegraphics[width=\textwidth]{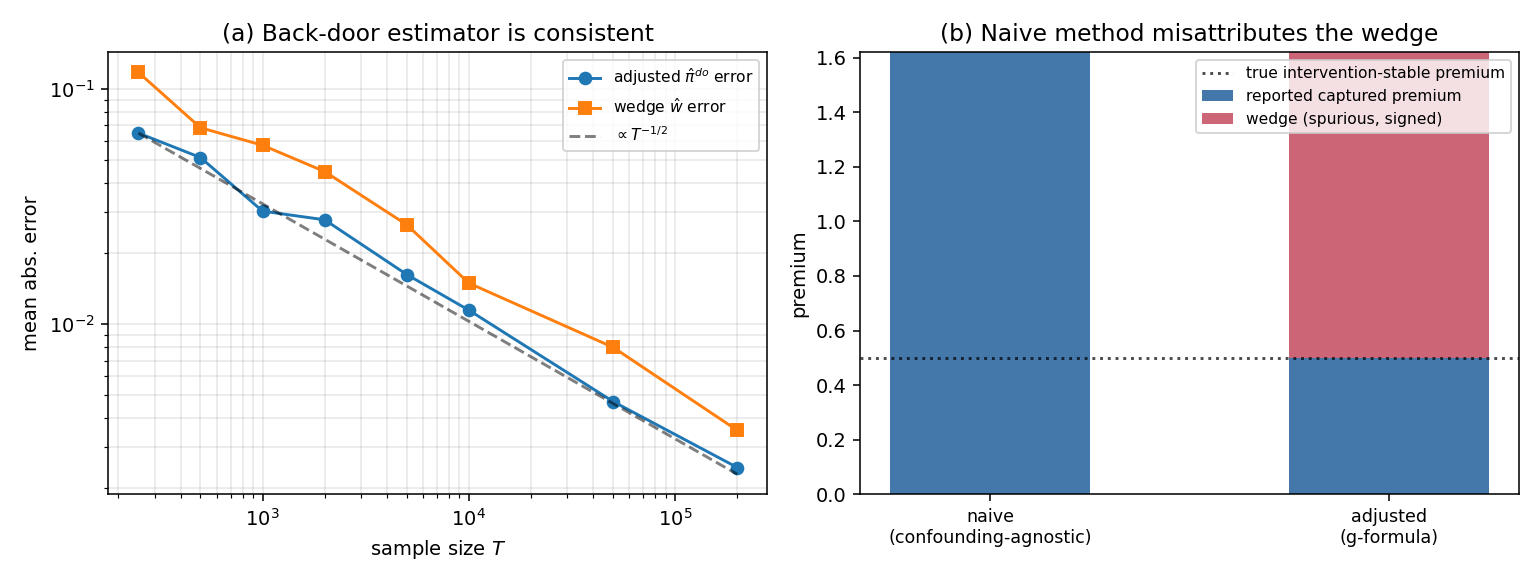}
\caption{Recovery. (a) The back-door estimator of $\pi^{\dop}$ and of the wedge is
consistent, with error $\propto T^{-1/2}$. (b) The confounding-agnostic method
reports the observational premium as captured, overstating the intervention-stable
premium by the wedge.}
\label{fig:recovery}
\end{figure}

\subsection{A multi-driver, multi-confounder panel}
The scalar picture extends to a realistic panel. We take $K=5$ observed drivers
loaded on $L=3$ latent confounders that also drive the price of risk, so the
observational attribution mixes the direct causal channel with confounding through
the shared latent factors. Here the population wedge is $+197\%$ of the
intervention-stable premium: a naive method would credit the drivers with roughly
three times the premium that survives intervention. Figure~\ref{fig:multidriver}(a)
shows the per-driver attribution under the naive and back-door methods---some
drivers are over-credited, and one even flips sign---and panel (b) confirms the
matrix g-formula estimator is consistent at the $T^{-1/2}$ rate. The decomposition
is therefore not a scalar curiosity: it operates on a covariance-matrix panel and
the causal correction matters most precisely when confounding is strong.

\begin{figure}[ht]
\centering
\includegraphics[width=\textwidth]{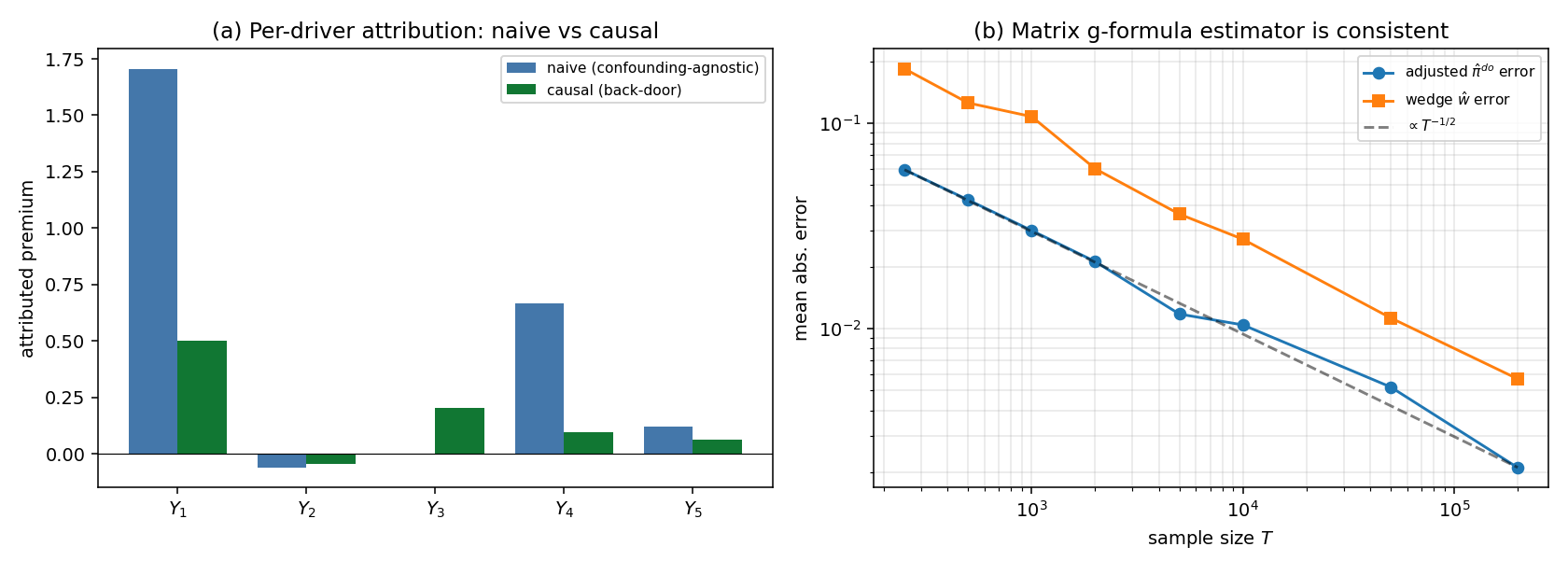}
\caption{Multi-driver, multi-confounder panel ($K=5$, $L=3$). (a) Per-driver
attributed premium: the confounding-agnostic method over-credits drivers and can
assign the wrong sign; the back-door method recovers the intervention-stable
credit. (b) The matrix g-formula estimator is consistent ($\propto T^{-1/2}$).}
\label{fig:multidriver}
\end{figure}

\subsection{Order-three obstruction and a screening protocol}
The order-three masking obstruction suggests a concrete screen for candidate driver
triples: (1)--(3) estimate the singleton, pairwise, and triple predictive $R^2$ of
the reference innovation; (4) form the incremental third-order signal
$\varepsilon=R^2(X\mid\varepsilon_1,\varepsilon_2,\varepsilon_3)
-\max_{i,j}R^2(X\mid\varepsilon_i,\varepsilon_j)$; (5) calibrate $\varepsilon$
against a permutation null that breaks the driver--future link, flagging the triple
when the increment exceeds the upper-$\alpha$ null quantile. All $R^2$ are estimated
out of sample (cell means fit on one half, scored on the other).

Figure~\ref{fig:order3} instantiates the population construction with
$X=B^1_T-B^1_{T/2}$, $Z=\sgn(X)$, independent Rademacher $\varepsilon_1,\varepsilon_2$,
$\varepsilon_3=\varepsilon_1\varepsilon_2 Z$: every singleton and pair has
out-of-sample $R^2$ indistinguishable from zero while the triple attains
$\approx2/\pi$ (A); the observed incremental order-three $R^2$ lies far outside the
permutation null ($200$ permutations, $p\approx0.005$) while the best pair sits
inside it (B); and the noise path, replacing the exact relation by
$X=\beta\,\sgn(\varepsilon_1\varepsilon_2\varepsilon_3)+\sigma\eta$, degrades the
incremental $R^2$ continuously while $\Delta=\Var(X)\,\varepsilon$ stays
constant (C), a finite-sample confirmation of Proposition~\ref{prop:approx}.

\begin{figure}[ht]
\centering
\includegraphics[width=\textwidth]{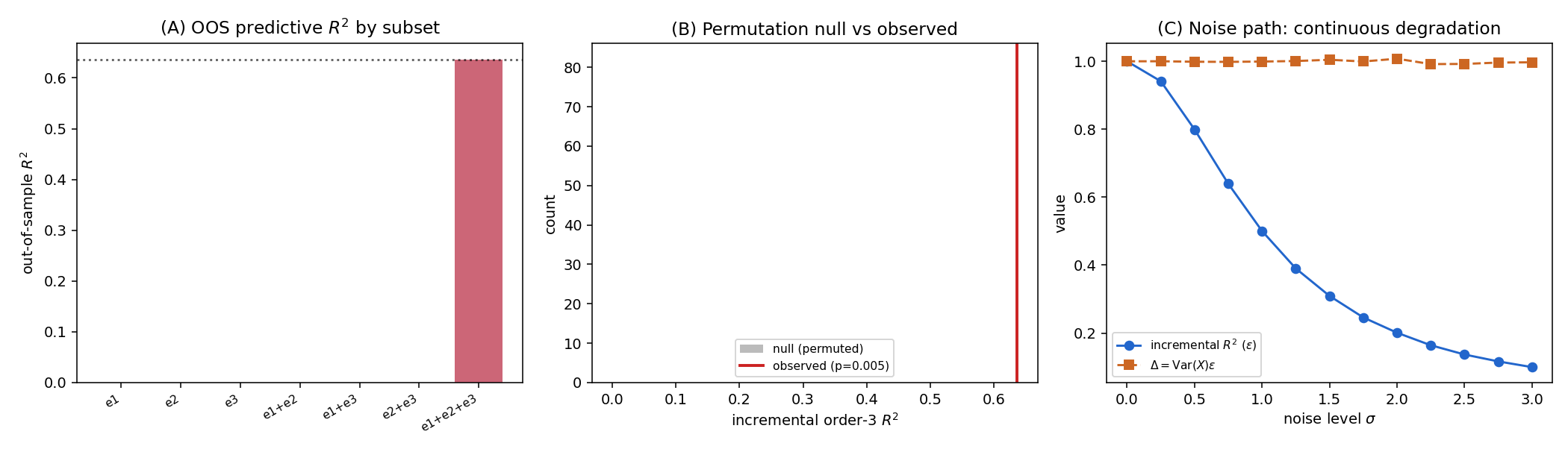}
\caption{Order-three obstruction. (A) Out-of-sample predictive $R^2$ by subset:
orders one and two are indistinguishable from zero, only the triple reaches
$2/\pi$. (B) Permutation null ($200$ permutations) versus the observed incremental
order-three $R^2$ ($p\approx0.005$). (C) Noise path: incremental $R^2$ degrades
continuously with noise while $\Delta=\Var(X)\varepsilon$ is invariant.}
\label{fig:order3}
\end{figure}

Finally we run the screen as a screen: one masking triple is planted among clean
drivers in a synthetic book, and we measure detection power and false positives.
In the synthetic design considered here, Figure~\ref{fig:screening}(a) shows the
screen separates the planted triple from clean triples and produces no false
positives among the sampled clean triples at the chosen permutation threshold; panel (b) traces
detection power against masking corruption $q$ (the fraction of the triple's
relation replaced by noise) at two sample sizes, exhibiting the expected
sample-size/signal-strength trade-off; panel (c) confirms the transition at an
intermediate sample size. The screen is reliable when the masking relation is close
to exact and the sample is adequate, and degrades gracefully otherwise.

\begin{figure}[ht]
\centering
\includegraphics[width=\textwidth]{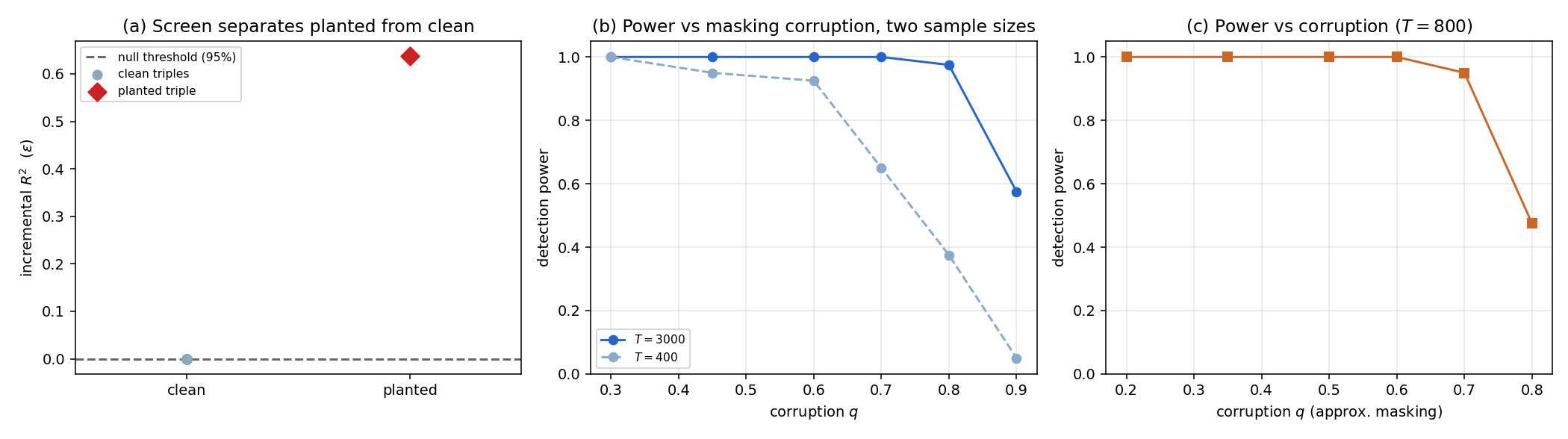}
\caption{Screening protocol on a book with one planted triple. (a) The incremental
$R^2$ separates the planted triple from clean triples; clean triples produce no
false positives against the $95\%$ null threshold. (b) Detection power versus
masking corruption $q$ at $T=400$ and $T=3000$. (c) Power versus corruption at
$T=800$.}
\label{fig:screening}
\end{figure}

\section{Scope and applicability}\label{sec:conclusion}

The price of risk, and through it the premium a causal strategy can realize, is a
function of the admissibility of the driver filtration. For a single portfolio the
conditional price-of-risk attribution decomposes exactly into an intervention-stable
component, a signed confounding wedge, and a nonnegative information loss, with the
loss an $L^2$ projection residual and the wedge a causal object that is not. This
decomposition need not aggregate across portfolios pooling their drivers: we exhibit
an obstruction at order three that is invisible to every singleton and pairwise
admissibility screen, at which the pooled filtration ceases to be admissible and the
projection algebra assigns value to information that anticipates future innovations.
The maximal jointly admissible sub-books are the independent sets of the masking
hypergraph.

Three limitations bound the claims. The wedge requires a specified causal graph and
a valid adjustment set; we do not solve causal discovery \cite{spirtes,peters,shimizu},
which is orthogonal and upstream to the attribution studied here. The experiments
demonstrate that the decomposition is estimable, that the causal correction is
necessary and recoverable, and that the order-three screen detects planted leakage
with controlled false positives; we do not, however, claim a particular empirical
frequency of the obstruction in real books---establishing that is a separate
empirical undertaking. And the obstruction requires two ingredients: a third-order
aggregate signal invisible to pairwise tests, and the coupling of that signal to a
future innovation. A purely adapted crowding mechanism supplies the first but not
the second, so it reproduces the masking algebra without the immersion failure
(Proposition~\ref{prop:crowding}); whether an adapted book can endogenously supply
both---generating a price filtration that reveals a future innovation---is left
open. The contribution is thus both structural and operational: it identifies the
object to look for, the order at which lower-order diagnostics are guaranteed to
miss it, and a concrete screen for it, while leaving its empirical prevalence to
future work.

\appendix
\section{Reproducibility details}\label{app:repro}

All experiments are synthetic and seeded with \texttt{numpy.random.default\_rng};
the parameters below fully specify each design, and the archived code
(\url{https://doi.org/10.5281/zenodo.20843643}) regenerates every figure.

\paragraph{Recovery (Figure~\ref{fig:recovery}).}
Single-driver linear--Gaussian model $U,Z\sim N(0,1)$ independent,
$Y=aU+\varepsilon_Y$ with $\varepsilon_Y\sim N(0,\sigma_Y^2)$, and
$\lambda=bY+cU+dZ$, with $(a,b,c,d)=(1.0,0.5,0.8,0.6)$ and $\sigma_Y^2=1$. The
adjustment set is $U$. The back-door estimator regresses $\lambda$ on $(Y,U)$ and
averages over the empirical marginal of $U$; the confounding-agnostic estimator
regresses $\lambda$ on $Y$ alone. Consistency is assessed over sample sizes
$T\in\{250,\dots,2\times10^5\}$ with $40$ replications per size.

\paragraph{Multi-driver panel (Figure~\ref{fig:multidriver}).}
$K=5$ drivers and $L=3$ latent confounders, with $Y=AU+E$,
$E\sim N(0,S)$ diagonal, and $\lambda=b^\top Y+c^\top U$; the loading matrix $A$,
the idiosyncratic variances $S$, and the coefficient vectors $b,c$ are fixed in the
script. Per-driver credit is $m_i(\Sigma_Y m)_i$ so the credits sum to the
respective premium. Consistency uses the same sample-size grid and $40$
replications.

\paragraph{Order-three obstruction (Figure~\ref{fig:order3}).}
$X=B^1_T-B^1_{T/2}\sim N(0,T/2)$ with $T=2$, $Z=\sgn X$, independent Rademacher
$\varepsilon_1,\varepsilon_2$, and $\varepsilon_3=\varepsilon_1\varepsilon_2 Z$,
with $N=2\times10^6$ samples. Predictive $R^2$ is estimated out of sample on a
$50/50$ split. The permutation null uses $200$ permutations of $X$ against the fixed
driver signs. The noise path replaces the exact relation by
$X=\beta\,\sgn(\varepsilon_1\varepsilon_2\varepsilon_3)+\sigma\eta$,
$\eta\sim N(0,1)$, over $\sigma\in[0,3]$ on a $13$-point grid.

\paragraph{Screening protocol (Figure~\ref{fig:screening}).}
A synthetic book of clean Rademacher drivers with one planted triple
$\varepsilon_3=\varepsilon_1\varepsilon_2 Z$, optionally corrupted by replacing a
fraction $q$ of the relation with independent signs. The screen flags a triple when
its out-of-sample incremental $R^2$ exceeds the upper-$5\%$ quantile of a
permutation null. Detection power is averaged over $40$ replications and traced
against $q\in[0.3,0.9]$ at sample sizes $T\in\{400,800,3000\}$.

\section*{Declarations}
\noindent\textbf{Funding} The author received no specific funding for this work.\\
\noindent\textbf{Conflict of interest} The author declares no conflict of interest.\\
\noindent\textbf{Data availability} All experiments are synthetic and use no
proprietary data. The code reproducing every figure is openly available at
\url{https://doi.org/10.5281/zenodo.20843643} (archived) and
\url{https://github.com/AlejandroRodriguezDominguez/order-three-attribution}
(development); each script is seeded, so every figure is exactly reproducible.


\begin{thebibliography}{99}
\bibitem{abs} A. Abhyankar, D. Basu, A. Stremme, \emph{The optimal use of
return predictability: an empirical study}, Journal of Financial and
Quantitative Analysis \textbf{47} (2012), no.~5, 973--1001.


\bibitem{adi} S. Ankirchner, S. Dereich, and P. Imkeller, \emph{The Shannon
information of filtrations and the additional logarithmic utility of insiders},
Ann. Probab. \textbf{34} (2006), 743--778.

\bibitem{admissible} A. Rodriguez Dominguez, \emph{Admissible information
structures, immersion, and the order of non-anticipative aggregation}, preprint,
2026 (superseding arXiv:2601.12541).

\bibitem{cpcm} A. Rodriguez Dominguez, \emph{Causal PDE-control models: a
structural framework for dynamic portfolio optimization}, preprint
arXiv:2509.09585, 2025.

\bibitem{fersonsiegel} W.~E. Ferson, A.~F. Siegel, \emph{The efficient use of
conditioning information in portfolios}, Journal of Finance \textbf{56} (2001),
967--982.

\bibitem{hansenrichard} L. P. Hansen and S. F. Richard, \emph{The role of
conditioning information in deducing testable restrictions implied by dynamic asset
pricing models}, Econometrica \textbf{55} (1987), no.~3, 587--613.

\bibitem{jacod} J. Jacod, \emph{Grossissement initial, hypoth\`ese (H'), et
th\'eor\`eme de Girsanov}, in S\'eminaire de Calcul Stochastique 1982/83, Lecture
Notes in Math. \textbf{1118}, Springer, 1985, 15--35.

\bibitem{kardaras} C. Kardaras, \emph{Market viability via absence of arbitrage of
the first kind}, Finance Stoch. \textbf{16} (2012), 651--667.

\bibitem{kk} I. Karatzas and C. Kardaras, \emph{The num\'eraire portfolio in
semimartingale financial models}, Finance Stoch. \textbf{11} (2007), 447--493.

\bibitem{mityagin} B. Mityagin, \emph{The zero set of a real analytic function},
Mathematical Notes \textbf{107} (2020), 529--530.

\bibitem{pearl} J. Pearl, \emph{Causality: Models, Reasoning, and Inference}, 2nd
ed., Cambridge University Press, 2009.

\bibitem{peters} J. Peters, P. B\"uhlmann, and N. Meinshausen, \emph{Causal
inference by using invariant prediction: identification and confidence intervals},
J. R. Stat. Soc. Ser. B \textbf{78} (2016), 947--1012.

\bibitem{protter} P. E. Protter, \emph{Stochastic Integration and Differential
Equations}, 2nd ed., Springer, 2004.

\bibitem{robins} J. Robins, \emph{A new approach to causal inference in mortality
studies with a sustained exposure period}, Math. Modelling \textbf{7} (1986),
1393--1512.

\bibitem{shimizu} S. Shimizu, P. O. Hoyer, A. Hyv\"arinen, and A. Kerminen,
\emph{A linear non-Gaussian acyclic model for causal discovery}, J. Mach. Learn.
Res. \textbf{7} (2006), 2003--2030.

\bibitem{spirtes} P. Spirtes, C. Glymour, and R. Scheines, \emph{Causation,
Prediction, and Search}, 2nd ed., MIT Press, 2000.
\end{thebibliography}
\end{document}